\documentclass[conference,a4paper]{IEEEtran}

\usepackage[utf8]{inputenc} 
\usepackage[T1]{fontenc}
\usepackage{url}              % provides \url{...}
\usepackage{cite}             % improves presentation of citations
\usepackage[cmex10]{amsmath}  % Use the [cmex10] option to ensure complicance
\usepackage{mleftright}       % fix to wrong spacing of \left-,
\mleftright                   % \middle- \right-commands 

\usepackage{graphicx}         % provides \includegraphics{...} to
\usepackage{booktabs}         % fixes poor spacing in tables and
\usepackage{subcaption}
\usepackage{tikz}
\usepackage{amsthm,amsfonts,amsmath,amssymb,bm}
\usepackage{algorithm,algorithmic}
\usepackage{setspace}
\usepackage{arydshln}
\usepackage{mathtools}
\usepackage{cleveref}
\usepackage{cite}
\usepackage{wrapfig}
\usepackage{bbm}
\usepackage{xfrac}
\usepackage{enumerate}
\usepackage[normalem]{ulem}
\usepackage{balance}
\usepackage{graphicx,soul}

\newtheorem{lemma}{Lemma}
\newtheorem{corollary}{Corollary}
\newtheorem{claim}{Claim}

\newcommand{\expect}{\mathbb E}
\newcommand{\prob}{\mathbb P}
\newcommand{\vct}[1]{\boldsymbol{#1}} % vector
\newcommand{\mat}[1]{\boldsymbol{#1}} % matrix
\DeclareMathOperator*{\argmax}{arg\,max}

\definecolor{pavlosred}{rgb}{0.7, 0.0, 0.0}
\newcommand{\pavlos}[1]{{\color{pavlosred}{Pavlos: #1}}}

\definecolor{ceruleanblue}{rgb}{0.16, 0.32, 0.75}
\newcommand{\chaorui}[1]{{\color{ceruleanblue}{Chaorui: #1}}}

\newcommand{\defVariable}{\mathit{U}}
\newcommand{\nTests}{\mathit{T}}
\newcommand{\testResult}{\mathit{Y}}
\newcommand{\testIndex}{\mathit{\tau}}
\newcommand{\nItems}{\mathit{n}}
\newcommand{\nDef}{\mathit{k}}
\newcommand{\defProb}{\mathit{p}}
\newcommand{\depth}{\mathit{d}}
\newcommand{\binEntropy}{\mathit{h_2}}

\IEEEoverridecommandlockouts
\begin{document}

\title{A Diagonal Splitting Algorithm \\for Adaptive Group Testing}
% \thanks{This work was supported, in part, by NSF grant #2146828. We also thank Katerina Argyraki for her valuable support.}
%%%%%%
% \author{ 
%    \IEEEauthorblockN{Anonymous Authors}
% \IEEEauthorblockA{%
%    Please do NOT provide authors' names and affiliations\\
%   in the paper submitted for review, but keep this placeholder.\\
%    ISIT23 follows a \textbf{double-blind reviewing policy}.}
% }

%%%%%% Please only add the author names and affiliations for the FINAL
%%%%%% version of the paper, but NOT for the paper submitted for review!
%
%%%%%
%%%%% Single author, or several authors with same affiliation:
% \author{%
%   \IEEEauthorblockN{Stefan M.~Moser}
%   \IEEEauthorblockA{ETH Zürich\\
%                     8092 Zürich, Switzerland\\
%                     moser@isi.ee.ethz.ch}
%                   }
%
%%%%%
%%%%% Several authors with up to three affiliations:
\author{%
\IEEEauthorblockN{Chaorui Yao}
  \IEEEauthorblockA{
                    University of California, Los Angeles,\\ 
                    Los Angeles, CA 90095, USA\\
                    chaorui@ucla.edu}
   \and                   
  \IEEEauthorblockN{Pavlos Nikolopoulos}
  \IEEEauthorblockA{EPFL\\
                    1015 Lausanne, Switzerland\\
                    pavlos.nikolopoulos@epfl.ch}

\and
  \IEEEauthorblockN{Christina Fragouli}
  \IEEEauthorblockA{
                    University of California, Los Angeles,\\ 
                    Los Angeles, CA 90095, USA\\
                    christina.fragouli@ucla.edu}
  
}

% \thanks{This work was supported, in part, by NSF grant #2146828. }

%
%%%%%   
%%%%% Many authors with many affiliations:
% \author{%
%   \IEEEauthorblockN{Albus Dumbledore\IEEEauthorrefmark{1},
%                     Olympe Maxime\IEEEauthorrefmark{2},
%                     Stefan M.~Moser\IEEEauthorrefmark{3}\IEEEauthorrefmark{4},
%                     and Harry Potter\IEEEauthorrefmark{1}}
%   \IEEEauthorblockA{\IEEEauthorrefmark{1}%
%                     Hogwarts School of Witchcraft and Wizardry,
%                     1714 Hogsmeade, Scotland,
%                     \{dumbledore, potter\}@hogwarts.edu}
%   \IEEEauthorblockA{\IEEEauthorrefmark{2}%
%                     Beauxbatons Academy of Magic,
%                     1290 Pyrénées, France,
%                     maxime@beauxbatons.fr}
%   \IEEEauthorblockA{\IEEEauthorrefmark{3}%
%                     ETH Zürich, ISI (D-ITET), ETH Zentrum, 
%                     CH-8092 Zürich, Switzerland,
%                     moser@isi.ee.ethz.ch}
%   \IEEEauthorblockA{\IEEEauthorrefmark{4}%
%                     National Yang Ming Chiao Tung University (NYCU), 
%                     Hsinchu, Taiwan,
%                     moser@isi.ee.ethz.ch}
% }
%

\maketitle

%%%%%
%% Abstract: 
%% If your paper is eligible for the student paper award, please add
%% the comment "THIS PAPER IS ELIGIBLE FOR THE STUDENT PAPER
%% AWARD." as a first line in the abstract. 
%% For the final version of the accepted paper, please do not forget
%% to remove this comment!
%%
\begin{abstract}

  %\pavlos{To be updated...} 
  Group testing enables to identify infected individuals in a population 
using a smaller number of tests than individual testing. 
To achieve this, group testing algorithms commonly assume knowledge of the number of infected individuals; nonadaptive and several adaptive algorithms fall in this category. Some adaptive algorithms, like binary splitting, operate without this assumption, but 
require a number of stages that may scale linearly with the size of the population.
In this paper, we contribute a new algorithm that enables a balance between the number of tests and the number of stages used,  and which we
term diagonal splitting algorithm (DSA).
Diagonal splitting, like binary splitting, does not require knowledge of the number of  infected individuals,  yet unlike binary splitting, is order-optimal w.r.t.
the expected number of tests it requires and is guaranteed to
succeed in a small number of stages that scales at most logarithmically with the size of the population. 
%Group testing enables to identify infected individuals in a population using a smaller number of tests than individual testing.  Perhaps the most quintessential example of adaptive group testing is the binary splitting algorithm~\cite{binSplitting}. In this work, we propose and analyze a variant of this classical algorithm that we term diagonal group testing. Diagonal group testing, like binary splitting, does not require knowledge of the number of  infected individuals,  yet unlike binary splitting, if the number of infected individuals is linear in the population size, does not require a linearly increasing number of tests. This benefit comes only at a moderate increase of the required number of tests in sparser regimes, where the number of infected individuals is small compared to the population size.
Numerical evaluations, for diagonal splitting and a hybrid approach we propose, support our theoretical findings. 
\end{abstract}

\iffalse
\section{Organization}
\begin{enumerate}
\item Introduction (1 page)
\item  Preliminaries on Group Testing and Prior Work (1 page)
\begin{enumerate}
        \item Group testing: Describe Combinatorial and Probabilistic infection models;
        \item Adaptive group testing: 
                -> Discuss the two well-known adaptive algorithms: Binary Splitting and Hwang’s Generalised Binary Splitting;
                -> Focus on the main concept behind Hwang’s optimal algorithm = split the total population into n/k groups. (This will help us intuitively support the idea of diagonal testing in the next section);
        \item Related work: Discuss prior work on unknown k (especially "Binomial Group-Testing with an Unknown Proportion of Defectives"). 
        \end{enumerate}
\item Diagonal Algorithm (2 pages)
\begin{enumerate}
        \item Description: Give tree representation and a pseudo code;
        \item Performance Analysis: 
                -> Compute the expected the number of tests for exact identification in both the combinatorial and the probabilistic setting;
                -> Compare against (Hwang's) binary splitting in the asymptotic regime (when n grows large);
                -> Use a figure to showcase the benefits and give some intuition about the hybrid optimization that follows.
        \item Optimizations: Discuss the hybrid approach.
        \end{enumerate}
\item Numerics (1 page)
\begin{enumerate}
        \item Combinatorial setting
        \item Probabilistic setting
        \end{enumerate}
\item Refs
\end{enumerate}
\fi

\section{Introduction}

The group testing (GT) problem dates back in the times of World War II and Dorfman's seminal work~\cite{Dorfman}. 
Simply stated, it assumes a population of individuals out of which some are infected, and the goal is to design efficient testing strategies to identify the infections from the test results.
The main component of group testing is the so-called \textit{pooled test}, which mixes together the diagnostic samples from multiple individuals, and (in its noiseless form) returns ``negative'', if all the considered individuals in the pool are healthy; otherwise, it returns ``positive'', if at least one of them is infected.

Needless to say, group testing does not apply only to medical testing---it has a large number of applications, but it is currently being ``re-discovered'' in the context of COVID-19~\cite{art1,art2,art4,Cov-GpTest-1,Cov-GpTest-2,kucirka2020-PCR}, with several countries (including US, EU states, and China) having already deployed some variant of it~\cite{GroupTest-implement1,GroupTest-implement2-FDA}. 
While the focus of this paper is also on medical testing, our approach is not bound to it in any sense. 

A large number of GT variations have been examined since Dorfman's paper (see~\cite{GroupTestingMonograph,GroupTestingBook,nested} and references therein); to the best of our knowledge, however, the following setting is still open and only a few sub-optimal solutions exist:
Consider a population size of $\nItems$ individuals and suppose that either $\nDef$ of them are infected (with all combinations $\binom{\nItems}{\nDef}$ being equiprobable), or everyone is infected i.i.d.\ with probability $\defProb$, but no information is available for $\nDef$ or $\defProb$, which are assumed to be completely \textit{unknown}. 

We explore this setting by being interested in questions such as: can we minimize the \textit{number of tests}, when the infection probability is completely unknown? Does it make sense to administer further testing to estimate it and use that estimate to administer new tests, or should we directly apply some well-established group testing algorithm (such as~\cite{binSplitting,sobel,hwang}) and expect small numbers of tests? do we need new algorithms?

In addition to the number of tests, we are also interested in the \textit{number of stages}. Adaptive algorithms achieve zero-error recovery of all infection statuses with few tests, but they typically use a large and variable number of stages.
% For example, Binary Splitting~\cite{binSplitting}, which applies some form of binary search, needs on average $O(\nDef \log_{2} \nItems)$ tests and the outcome of each test should be available in order to administer the next one; hence, the number of stages is also $O(\nDef \log_{2} \nItems)$, where $\nDef$ denotes the number of infected individuals.
On the other hand, non-adaptive group testing can achieve zero-error recovery in one testing stage, but it typically requires the knowledge of $\nDef$ to be efficient and tends to ``blow up'' the number of tests as the number of infections grows large; in fact, 
classic individual testing has been proved to be asymptotically optimal among non-adaptive designs in the linear ($\nDef = \Theta(\nItems)$)~\cite{individual-optimal} and the mildly sublinear infection regime ($\nDef = \omega(\frac{\nItems}{\log\nItems})$)~\cite{bay2020optimal}. 
We ask: can we achieve \textit{exact recovery} with a smaller number of stages, without increasing the expected number of tests?

The practical application of such a setting is straightforward.
Suppose we want to test a population that has been exposed to a new disease for epidemiological purposes: e.g., to update known epidemiological models (such as \cite{mathOfEpidemicsOnNetworks}) or come up with new ones based on the test outcomes.
We want our testing to be both fast and accurate;
%we typically want to identify quickly and accurately the infection statuses of each individual,
%so that epidemiologists can better examine the time dynamics of the disease spread;
so, 
erroneous recovery and/or a large number of testing stages, especially if each stage may last more than a day (as with molecular PCR tests), are out of the question. 
Another application lies in a recent line of work~\cite{GroupTesting-community-nonOverlap,GroupTesting-community-overlap,ayfer2021adaptive,srinivasavaradhan2021dynamic,isit-paper,zhu2020noisy} that examines the group testing problem under the assumption that infections are governed by community spread, hence individuals that belong to the same community groups~\cite{GroupTesting-community-nonOverlap,GroupTesting-community-overlap,ayfer2021adaptive,srinivasavaradhan2021dynamic,isit-paper} 
or have experienced the same social interactions~\cite{zhu2020noisy} are assumed to be infected with approximately equal probability. 
That probability, however, can be hardly known in new viral-disease cases, because little is known about how a new virus propagates inside the community, and even if an arguably (in)accurate epidemiological model is used, the information about the initial state and the exact epidemiological trajectory of the disease is missing.
So, one is left with multiple group testing strategies, without knowing which one to use. 

In this regard, we contribute a new algorithm, which we call diagonal splitting algorithm (DSA), which is order-optimal w.r.t.\ the expected number of tests it requires and is guaranteed to succeed in a small number of stages that equals $\log_{2}{\nItems}$ in the worst case.
Our algorithm is a noiseless adaptive group testing algorithm and leverages the following  observation in order to reduce the number of tests and stages: 
Consider the binary tree representation of testing via binary search (akin to Binary splitting); we can test across ``diagonals'' 
% \chaorui{we will explain in the sequel} 
in this tree, so that, at each stage, we test in parallel all possibly infected members of the population, but partitioned in groups of variable sizes. This enables a balance between the number of tests and stages used,   performing well across both dimensions.

The paper is organized as follows: we first give some background and discuss related work (\Cref{sec:prelim}), then we describe and analyze DSA (\Cref{sec:algorithm}) and finally evaluate its performance (\Cref{sec:numerics}).

\section{Preliminaries}
\label{sec:prelim}
\subsection{Infection model and Problem formulation}

Three infection models are usually studied in the group testing literature: 
(I) in the ``combinatorial-priors model,''  a fixed number of  individuals $\nDef$ (selected uniformly at random), are infected; 
(II) in ``i.i.d probabilistic-priors model,'' each individual is i.i.d infected with some probability $\defProb$, so that the expected number of infected members is $\bar{\nDef} = \nItems \defProb$;
(III) in the ``non-identical probabilistic-priors model,'' each item $i$ is infected independently of all others with prior probability $\defProb_i$. 
Models (I) and (II) have received attention from researchers for the most part (see for example, \cite{bernoulli_testing2,bernoulli_testing1,Nonadaptive-1,PhaseTrans-SODA16,ncc-Johnson,individual-optimal,bay2020optimal,coja-oghlan20a, armendariz2020group,price2020fast,coja-oghlan19}),
while model (III) is the least studied one \cite{prior,gradient_GT}. 
We refer the reader to \cite{GroupTestingMonograph} for an excellent summary of existing work---the related parts to our work are described in the next section.

In this paper, we consider the noiseless adaptive group testing problem in the first two settings, where either $\nDef$ (in the combinatorial setting (I)) or $\defProb$ (in probabilistic setting (II)) are \textit{unknown}. 
However, our algorithm(s) can be used and the results of our analysis can be extended to model (III), as well.

Our goal is to design an algorithm that can identify the infection status of all individuals by using both a small number of tests and a small number of stages. 
This \textit{zero-error} recovery requirement is achievable with noiseless adaptive group testing (see next section), but well-known algorithms for doing so typically use a variable number of tests and stages. 
So, to fairly compare our proposed approaches against them, we use as performance metrics both: the expected numbers of tests and the expected number of stages. 

\subsection{Background and Related work}
% We next provide some background and discuss related work in our context.

A pooled test indexed by $\testIndex$ takes as input samples from a set of individuals  $\delta_{\testIndex}$ and outputs a binary value:  $1$ (``positive'') if at least one of the samples is infected, and $0$ (``negative'') if none of them is infected.
More precisely, let ${\defVariable_i=1}$  if individual $i$ is infected, and $0$ otherwise. 
The output of pooled test $\tau$ is calculated as $\testResult_\tau= \bigvee_{i\in \delta_{\testIndex}} \defVariable_i$, 
where $\bigvee$ stands for the \texttt{OR} operator (disjunction). 

\textbf{Minimum number of tests:}
In the combinatorial\footnote{The achievability/converse results are usually proved for combinatorial model (I) (a summary is in~\cite{price2020fast}), but they are directly applicable to model (II) by considering $\defProb=\nDef/\nItems$ (see Theorems 1.7 and 1.8 in \cite{GroupTestingMonograph} or~\cite{bay2020optimal}).}  model (I), since $\nTests$ tests allow to distinguish among $2^\nTests$ combinations of test outputs, 
we need 
$\nTests \geq \log_2{\binom{\nItems}{\nDef}}$ to identify $\nDef$ randomly infected individuals out of $\nItems$.
This is known as the \textit{counting  bound} and implies that in a sparse regime (i.e. $\nDef=\Theta(\nItems^\alpha)$ and $\alpha \in [0,1)$), no algorithm can use less than $O(\nDef\log \frac{\nItems}{\nDef})$ tests to achieve zero-error identification~\cite{CntBnd,GroupTestingBook}. 
In the probabilistic model (II), a similar bound exists for the number of tests needed on average: $\nTests \geq \nItems \binEntropy\left(\defProb\right)$,
where $\binEntropy$ is the binary entropy function~\cite{GroupTestingMonograph}.

\textbf{Adaptive GT:} 
Adaptive algorithms utilize the outcome of previous tests to determine what tests to perform next. They achieve zero-error recovery using a small but variable number of tests, and they typically use a large number of stages.

The commonest example is the {\em binary splitting algorithm (BSA)}, which implements a form of binary search (see~\cite[Procedure R4]{binSplitting} or~\cite[Algorithm 1.1]{GroupTestingMonograph}.
BSA was the first to introduce the idea of recursively splitting the population, but it failed to match the counting bound for any $\nDef$, as it is guaranteed to succeed in less than $\nDef \log_{2} \nItems + \nDef$ tests~\cite{capacity-adaptive}.
Also, as the outcome of a test must be known in order to administer the next one and no tests can be conducted in parallel, BSA needs on average as many stages as tests, i.e. linear in $\nDef$.

Another example is \textit{Hwang’s generalized binary splitting algorithm (HGBSA)}~\cite{hwang}, which also uses binary search and is already known to be order-optimal in terms of the number of tests.
% HGBSA uses various tweaks to accomplish that:
% For example, if $\nDef \geq \frac{\nItems}{2} + 1$, then HGBSA uses one stage of individual testing and avoids applying binary search, which would need more tests.
% Also, if $\nDef \leq \frac{\nItems}{2}$, the set sizes at each stage are chosen to be powers of $2$, so that the binary-splitting step is always exact; and most importantly, the items appearing in negative tests at some stage are not considered in subsequent testing stages.
% As a result, given known $\nDef$, HGBSA is guaranteed to succeed in
Given known $\nDef$, HGBSA is guaranteed to succeed in
$T = \log_{2} \binom{\nItems}{\nDef} + \nDef$ tests~\cite{hwang,capacity-adaptive}, and the expected numbers of stages is also reduced compared with BSA. 
HGBSA can be also applied with only an upper bound of $\nDef$; however, its large benefits are achievable, only if the exact $\nDef$ value is known. 

The most closely related work in our context of unknown $\nDef$ is~\cite{sobel}, which we will call \textit{Sobel's Binary Splitting algorithm (SBSA)}, for simplicity, after one of the authors of~\cite{sobel} (same as in~\cite{binSplitting}).
SBSA is similar to BSA, but does not perform binary search blindly; instead, at each stage the algorithm also takes into account the results of previous positive/negative tests and selects the size of the next pooled test accordingly. 
% Also, like HGBSA, individuals that are deemed healthy at any stage are removed from the population and only the remaining individuals are considered in subsequent stages. 

Although SBSA has been empirically shown to require few tests when $\nDef$ is unknown, it bears several weaknesses with regard to our (more practical) setting: First, according to the tables provided in~\cite{sobel}, the worst-case number of stages happens when everyone is infected (i.e. $\nDef = \nItems$), and it scales as $O(\nItems)$. 
% This is because SBSA has the same property as BSA; only one test is conducted at each stage and the algorithm needs to know the test outcome at one stage to proceed to the following test/stage---hence, no tests can be conducted in parallel. 
Second, at each stage, the algorithm chooses the size $X$ of the next pooled test based on some recursive expressions that depend on 
% the number of individuals that have been definitely classified as healthy $H$ and infected $I$ up to that stage, the size of the population under testing at that stage $N$, the size of the previous positive pooled test $M$, and a prior distribution for $\nDef$. Unfortunately, as any of $(H,I,N,M)$ may take a value close to $\nItems$, the ``on-line'' recursive computation becomes quickly intractable. 
% Moreover, an ``off-line'' pre-computation of all $X$ values for any possible tuple $(H,I,N,M)$ is quite expensive, scaling as $O(\nItems^4)$, for each prior distribution about $\nDef$ that is considered.
four parameters that may take any value close to $\nItems$;
and as a result, the ``on-line'' recursive computation becomes quickly intractable, while an ``off-line'' pre-computation of all $X$ values for any possible quadruple of parameters is quite expensive, scaling as $O(\nItems^4)$.
%, for each prior distribution about $\nDef$ that is considered.
Last, the recursive expressions in~\cite{sobel} cannot help provide theoretical guarantees on the average number of tests or characterize SBSA's asymptotic behavior. 

\textbf{Few-stage adaptive GT algorithms:}
Prior works on adaptive GT with few stages and a near-optimal number of tests do exist, but none of them is directly applicable to our context.
Some of the algorithms proposed require prior knowledge of the exact value of $\nDef$ (or $\defProb$) to achieve a small number of tests (e.g. see~\cite{Damaschke2009CompetitiveGT,aldridge2022conservative} and Section 5.2 of \cite{GroupTestingMonograph}), but this violates our assumption of unknown $\nDef$---of course, one might obtain an estimate of $\nDef$ (e.g. by using approaches from \cite{Damaschke2009CompetitiveGT} or~\cite{FJOPSP16}), but this would require initially even more  tests and/or stages.
Other algorithms can only guarantee a limited number of stages under the PAC (instead of zero-error) recovery criterion~\cite{Bonis2005OptimalTA}.
Finally, as mentioned in Section 5.3 of \cite{GroupTestingMonograph}, which offers a survey on GT ``universality,'' i.e. the case where no prior knowledge is assumed, the question of the optimal order of required tests for exact recovery is still open.

\section{Diagonal splitting algorithm (DSA)}
\subsection{Algorithm}
\label{sec:algorithm}
\begin{figure}[t]
\begin{center}
\includegraphics[height = 3.5cm, width=0.45\textwidth]{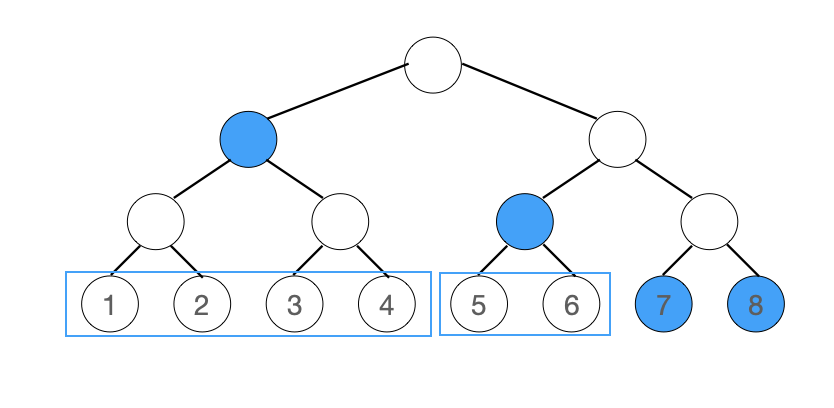}
\end{center}
\caption{An illustration of DSA for $\nItems = 8$. The blue nodes correspond to the pooled tests of the first stage.}
\label{fig:diag-demo}
\end{figure}

Suppose that $\nItems  = 2^\depth$, $\depth \in \{1,2,\ldots\}$ and consider the binary tree representation of testing via binary search. 
\Cref{fig:diag-demo} depicts an example of a population of $8$ individuals.
Each leaf of the tree corresponds to an individual of the group that may be ``negative'' (i.e., healthy) or ``positive'' (i.e., infected).
Following the standard terminology of pooled testing, at each particular level of the binary tree, a node is considered to be ``positive'', if at least one of the leaves of the sub-tree rooted from that node is positive, and it is ``negative'' otherwise.

\Cref{alg:diag} applies pooled testing only across diagonals in this tree:
At the first stage (line 1), we select $d-1$ (internal) nodes, 
starting from the second level and going down to the last-but-one level (of the leaves), so that: 
(a) no two nodes are at the same level, 
and (b) all the sub-trees rooted from the selected nodes are disjoint. 
Then, we conduct a total of $d-1$ pooled tests; each pooled test corresponds to a particular node and contains the diagnostic samples of all the leaves of the sub-tree that is rooted from it.
The outcome of each pooled test is also the label (positive or negative) of the corresponding node. 
At the same (first) stage, we also test individually the only two leaves that do not belong to any sub-tree, which gives a total of $d+1$ tests. 
See, for example, the blue nodes in the example of~\Cref{fig:diag-demo}.  

At each subsequent stage (lines 2-3), we just apply the same diagonal pooled testing procedure in parallel across all sub-trees whose root node was found positive in the current stage.  {In the example of Figure~\ref{fig:diag-demo}, if all test outcomes return positive in stage one, in stage two we perform 5 tests in parallel, a group test for leaves 1 and 2, and individual tests for leaves 3-6. An extended example is provided in \Cref{app:tree16} of~\cite{yao2023diagonal}.} 

\begin{algorithm}[t]
\caption{DSA (multi-stage).}
\begin{algorithmic}[1]
\label{alg:diag}
\STATE Apply diagonal pooled tests to the set $A$. \\
    \STATE If all the preceding results are negative, A contains no defective items, and we halt. Otherwise, continue.\\
    \STATE If any preceding pooling consists of a single item, then that item is defective. Otherwise, apply diagonal pooled tests to all positive sub-trees.
\end{algorithmic}
\end{algorithm}

The key concept behind DSA follows from a variant of HGBSA that is described in~\cite{GroupTestingMonograph}: If one knows the exact value of $\nDef$, they can reduce both the number of tests and stages by first splitting the population into $\nDef$ subsets. 
Since each of these subsets contains an average of one infected individual, if we use BSA in each subset, we can find all infections using no more than $\nDef \log_{2}(\nItems/\nDef)+O(\nDef)$~\cite{GroupTestingMonograph} tests, which is generally smaller than what is achieved by BSA and SBSA. 

Similarly, the main benefit of DSA is that: at
each stage, we test in parallel all possibly infected members
of the population, but partitioned into groups of variable sizes.
This observation is critical for our work; as we show next, this partitioning enables finding a good balance between the expected number of tests and the expected number of stages. 
 
\subsection{Analysis}
\label{sec:analysis}
Given that infections follow either the combinatorial (I) or probabilistic (II) model, we next analyze the performance of \Cref{alg:diag} w.r.t.\ the expected number of tests and stages. For the sake of simplicity of the computations, we will assume that an initial pooled test of the entire population has already been found positive before \Cref{alg:diag} is called.

We first list a few properties of the tree structure:
\begin{enumerate}[(1)]
\item At each depth $i \in \{1,2,\ldots,d-1\}$, only $2^{i-1}$ nodes can be possibly tested by \Cref{alg:diag}.
\item Any sub-tree rooted from a node at depth $i$ has $\nItems_i = 2^{d-i}$ leaves.
\item At any stage of \Cref{alg:diag}, if a node is found positive at depth $i$, with $i \in \{0,\ldots,d-1\}$, then this triggers $\nTests_i = d-i+1$ new pooled tests at the next stage. 
\item When a leaf is identified by \Cref{alg:diag}, so is its sibling---hence, both need the same number of tests. 
\end{enumerate}

\begin{lemma}
\label{lem:tests}
The expected number of tests required by \Cref{alg:diag} for exact recovery is: 
\begin{equation}
\label{eq:tests}
\expect[\nTests] = d+1 + \sum_{i=1}^{d-1} 2^{i-1} \cdot  \prob_{i,*}^{+} \cdot \left(d-i+1\right),
\end{equation}
where $\prob_{i,*}^{+}$ is the probability that a node at depth $i$ is found positive  and depends on the underlying infection model; i.e.:
\begin{equation}
\label{eq:prob.combinatorial}
     \prob_{i,Comb}^{+} = 1 - \frac{\binom{\nItems-\nDef}{\nItems_i}}{\binom{\nItems}{\nItems_i}}, \qquad
     \prob_{i,Prob}^{+} = 1- (1-\defProb)^{\nItems_i}  .  
\end{equation}
\end{lemma}
\begin{proof}
Note that the expected number of tests is equal to: $\expect[T] = d+1 + \sum_{i=1}^{d-1}\nTests_i \cdot \expect[\text{\# positive nodes at depth } i]$, where $T_i = d-i+1$ by property (3). 

Consider the nodes that can be tested by \Cref{alg:diag} at some depth $i$ and observe that the probability of any such node being positive is the same for all the nodes at the same depth. Denote that common probability by $\prob^{+}_{i, *}$; this equals the probability of at least one of the leaves of the sub-tree rooted from a node at depth $i$ being infected. For a combinatorial model, $\prob^{+}_{i, *}$ is computed with the help of the hypergeometric distribution, as: $\prob^{+}_{i, Comb} = 1- \sfrac{\binom{\nItems-k}{\nItems_i}}{\binom{\nItems}{\nItems_i}}$. For a probabilistic model with infection probability of $p$, $\prob^{+}_{i, *}$ is computed with the help of the binomial distribution, as: $\prob^{+}_{i, Prob} = 1- (1-p)^{\nItems_i}$. 
Using the above and property (1), we have: $\expect[\text{\# positive nodes at depth } i] = 2^{i-1} \cdot \prob^{+}_{i, *}$. 

Combining all the above concludes the proof.
\end{proof}

\begin{corollary}
\label{cor:k=1}
If $\nDef = 1$ or $\defProb = \sfrac{1}{\nItems}$, then
% \begin{equation}
$\expect[\nTests] = \frac{1}{4} d^2 + \frac{5}{4} d + \frac{1}{2}$.
% \end{equation}
\end{corollary}
%\begin{proof}
% In the Appendix.
% \end{proof}

\begin{corollary}
\label{cor:k=n}
If $\nDef = \nItems$ or $\defProb = 1$, then
% \begin{equation}
$\expect[\nTests] = \frac{3}{2} \nItems - 1$.
% \end{equation}
\end{corollary}

A few observations w.r.t.\ the expected number of tests:
\begin{itemize}
\item The expected number of tests of DSA in the combinatorial and the probabilistic model have the same asymptotic behavior (as $\nItems$ goes to infinity), if $\defProb = \sfrac{\nDef}{\nItems}$.
\item If $\nDef = 1$ (or $\defProb = \sfrac{1}{\nItems}$), DSA is slightly worse than BSA, which needs $\log_{2}{\nItems}+1$ tests on average. If $\nDef = \nItems$ (or $\defProb = 1$), DSA compares favorably to BSA, which needs $\nItems (\log_{2}\nItems + 1)$ tests on average (see \Cref{sec:prelim} for details).
\item In the case of arbitrary fixed $\nDef$ (or $\defProb$), DSA is asymptotically order-optimal w.r.t.\ the performance of HGBSA (and the counting bound), albeit with a possibly large constant. This claim is proved using a loose upper bound in \Cref{app:optimality} of~\cite{yao2023diagonal}, and it is also validated in our numerics (e.g. see Fig.~\ref{fig:comb1024:tests}).
\end{itemize}

\begin{lemma}
    The number of stages required by \Cref{alg:diag}, in the worst case, equals $\log_{2}\nItems$.
\end{lemma}
\begin{proof}
    The worst case for \Cref{alg:diag} is when (at least) one of the two leftmost leaves is infected, which results in a number of stages equal to the depth of the tree $d = \log_{2}\nItems$. 
\end{proof}

An important observation that follows from the above is that the worst-case number of stages of \Cref{alg:diag} becomes significantly smaller than the expected number of stages of BSA (which scales as $O(\nDef \log_{2}\nItems)$), as the number of infections $\nDef$ grows large (see also Fig.~\ref{fig:comb1024:stages} for numerical evidence).

It is worth repeating that DSA  achieves all the above without knowing anything about $\nDef$ (or $\defProb$).

\subsection{A hybrid variant}

We observe that the DSA neither assumes knowledge of the number of infections $\nDef$, nor attempts to estimate it. However,  even from the outcome of the  first stage of diagonal tests, we may be able to get an estimate of $\nDef$ that we can then use to further reduce the number of tests needed. For example, if all test outcomes at the first stage return positive, then with a high probability more than $30\%$ of the population size is infected, and thus we may directly proceed with individual testing (according to~\cite{LinBndIndv2}). 

In this section, we propose a simple hybrid algorithm that builds on this observation. We merge the diagonal splitting algorithm with HGBSA as described in Algorithm~\ref{alg:hybrid}. Recall that HGBSA decides the size of group tests depending on  $\nDef$, and reduces to individual testing for large $\nDef$. To estimate $\nDef$, we  simply use the test outcomes of the first stage of DSA. We do so via likelihood calculations: for different numbers of infections $\nDef$,  certain test outcome patterns $\vct{s}$ of the diagonal algorithm are more likely to occur.

%We hybrid the diagonal algorithm and Hwang's algorithm to optimize the diagonal algorithm. For different number of infections $\nDef$, certain outcome pattern the diagonal algorithm and Hwang's algorithm is more likely to occur, which means that we can apply the diagonal testing to estimate the number of infections $\nDef$. Then we will use this estimate as the upper bound of $\nDef$ in Hwang's algorithm. We propose the following decoder to estimate k. 

Given the population size $\nItems$, we first construct a $(\nItems+1)$-by-$(2\nItems)$ matrix $\mat{M}$, where the rows correspond to the possible number of defective items $k$, the columns correspond to possible outcome pattern of the first/upper diagonal slicing, and each element $\mat{M}[k,\vct{s}]$ counts the possible occurrences of a test outcome pattern $\vct{s}$ for a given number of infections $\nDef$---we represent the test outcomes with binary vectors, hence the bold font of $\vct{s}$. 
% Each element $\mat{M}[k,s]$ is the occurrence of pattern $\vct{s}$ when the number of infections is $\nDef$. 
% We see a pattern $\vct{s}$ is valid for $\nDef$ infections when it is possible for the pattern $\vct{s}$ to occur given $\nDef$ infections. 
% Denote by $u_{\vct{s}}$ the maximum number of infections for a given pattern $\vct{s}$.
Let $\vct{s_0}$ denote the pattern, where we modify the first nonzero bit of $\vct{s}$ to zero. Also, denote by $N_0$ the number of leaves in the sub-tree, where the first nonzero bit  of $\vct{s}$ refers to. 
The following lemma enables us to efficiently populate the table~$\mat{M}$. Its proof is given in \Cref{app:recursiveM} of~\cite{yao2023diagonal}; an example such matrix is in \Cref{app:exampleM} of~\cite{yao2023diagonal}.
\begin{lemma}
\label{lem:recursiveM}
The elements of $\mat{M}$ have the following recurrence:
\begin{equation}
   \mat{M}[\nDef,\vct{s}] = 
    \sum_{i=1}^{\nDef}\binom{N_0}{i}\mat{M}[\nDef-i,\vct{s_0}], \text{ where } \mat{M}[0,0]=1.
\end{equation}
\end{lemma}

Next, we normalize $\mat{M}$ by row to get the 
% probability of pattern for a given number of infections $\nDef$, $\mathbb{P}[s|k]$. Denote by $\mat{M}_{s|k}$. Then the 
likelihood values for each $\nDef$: $L_{\vct{s}}(\nDef) = \prob[\vct{s}|k]$. 

\begin{algorithm}[t]
\caption{Hybrid GT (multi-stage).}
\begin{algorithmic}[1]
\label{alg:hybrid}
\STATE Apply diagonal pooled tests to the set $A$. Get pattern $\vct{s}$. \\
    \STATE Estimate $\hat{\nDef} = \argmax_\nDef L_{\vct{s}}(\nDef)$. Allocate $\hat{\nDef}$ to each positive sub-trees proportionally to the number of their leaves.\\
    \STATE Use HGBSA in positive sub-trees with allocated estimates.
\end{algorithmic}
\end{algorithm}

Our hybrid variant uses the normalized version of $\mat{M}$ and \Cref{alg:hybrid} to identify the status of every individual:
it first applies one stage of diagonal pooled tests to obtain the test outcome pattern $\vct{s}$ (line 1).
Given $\vct{s}$, it estimates the total number of estimated infections $\nDef$ as: $\hat{\nDef} = \argmax_\nDef L_{\vct{s}}(\nDef)$, and it allocates $\hat{\nDef}$ to each positive-identified sub-tree proportionally to the number of its leaves (line 2).
Finally, in each positive sub-tree, it runs HGBSA assuming that the number of infections equals the allocated estimate (line 3).

% We can design a maximum likelihood estimator as a decoder. Denote the maximum likelihood decoder as $D$. We have the following hybrid algorithm.

\section{Numerical Evaluation}
\label{sec:numerics}
\begin{figure}[t]
\label{fig:comb16}
\begin{subfigure}{\linewidth}
\centering
\includegraphics[height = 5cm, width = 0.9\linewidth]{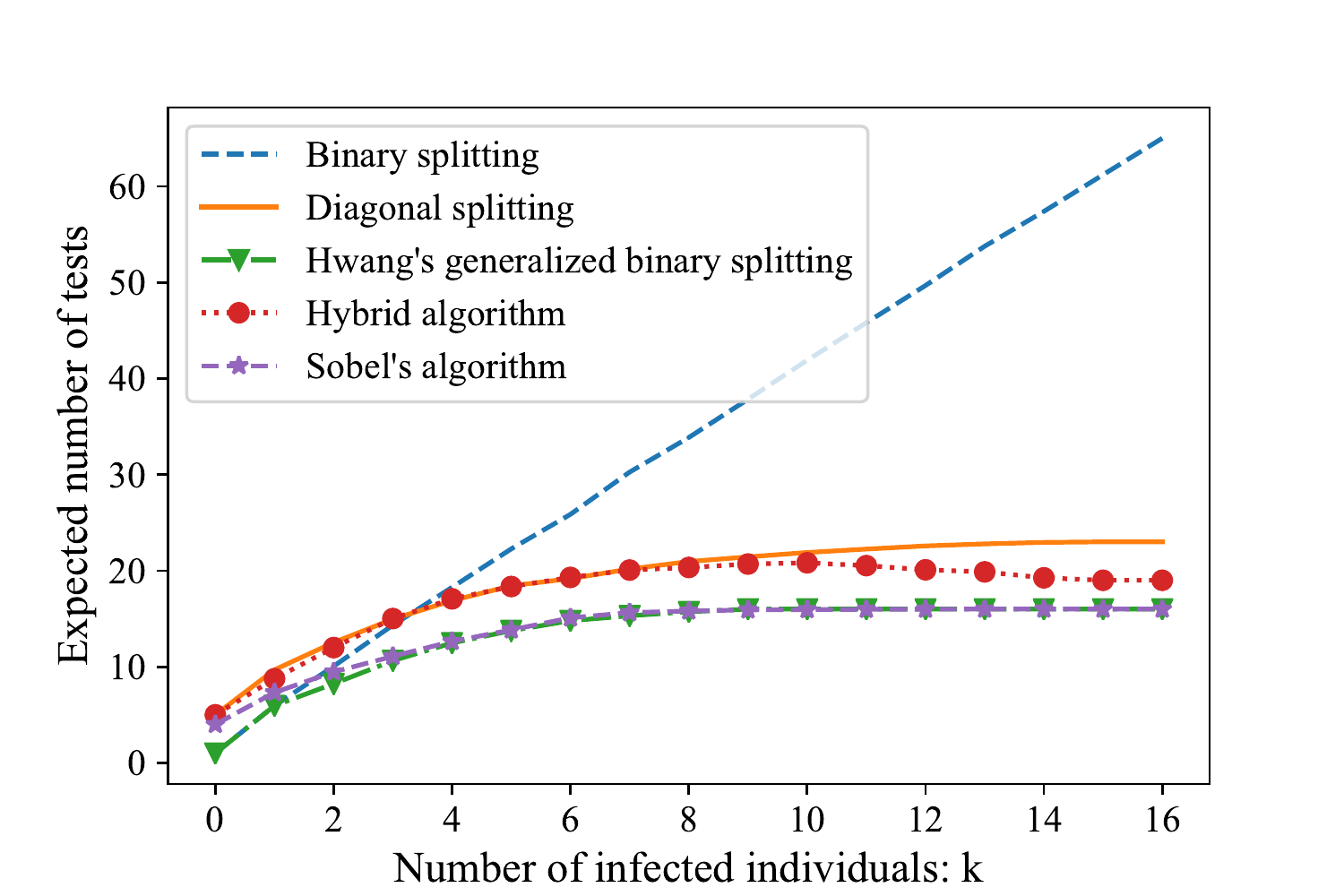}
\caption{Number of tests}
\label{fig:comb16:tests}
\end{subfigure}

\begin{subfigure}{\linewidth}
\centering
\includegraphics[height = 5cm, width = 0.9\linewidth]{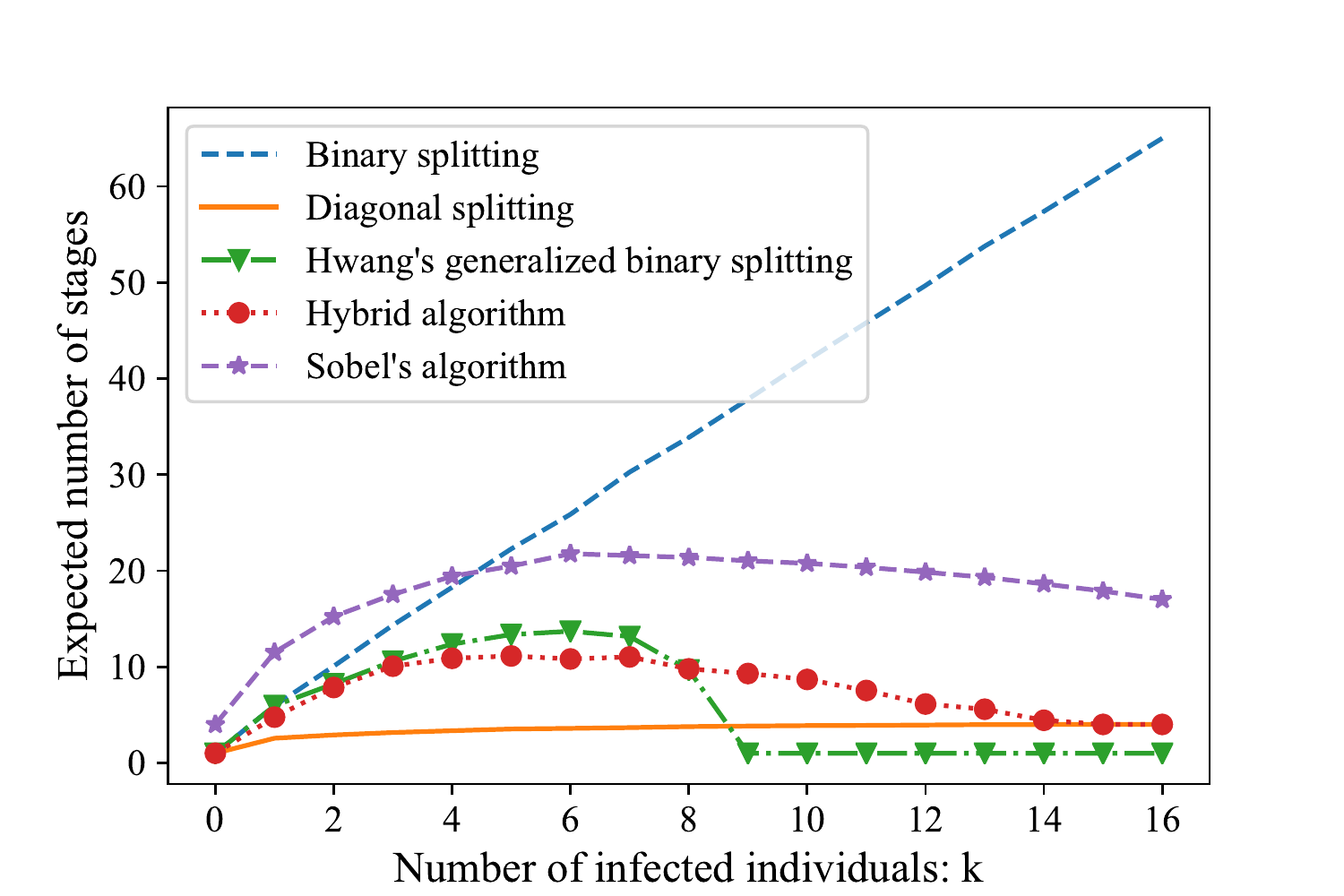}
\caption{Number of stages}
\label{fig:comb16:stages}
\end{subfigure}%
\caption{Case $\nItems=16$, results averaged over $500$ instances.}
\vspace{-0.5cm}
\end{figure}

\textit{Experimental setup:}
In various populations of two different sizes ($\nItems = 16$ and $\nItems = 1024$), some individuals are randomly infected according to the combinatorial and probabilistic infection models from \Cref{sec:prelim}. 
We use BSA, HGBSA, SBSA and our Algorithms \ref{alg:diag} and \ref{alg:hybrid} to identify the infected individuals without error, and each time we evaluate their performance in terms of both the total number of tests and stages they need. 
For brevity, we report only the combinatorial case---the probabilistic case is in \Cref{app:probabilistic} of~\cite{yao2023diagonal} with similar results. 

\textit{Notes:} (a) HGBSA does leverage the exact knowledge of $\nDef$ or $\defProb$, while the other algorithms are agnostic to them;
(b) SBSA was proved to be computationally expensive, and we could only evaluate its performance in the case of $\nItems = 16$.

\textit{Combinatorial setting:}
To examine the performance of the algorithms w.r.t.\ various numbers of infected individuals $\nDef$, we range $\nDef$ in $\{1,...,\nItems\}$. 
For each integer value of $\nDef$, we generate at least $500$ random instances of infected populations over which we average our results. Following the combinatorial infection model in each instance, we randomly select $\nDef$ people to be infected, so that all $\binom{\nItems}{\nDef}$ combinations are equiprobable.

\iffalse
\chaorui{To be updated. The simulation result of 1024 is not ready yet :(. Current version serves as a placeholder.}
\label{sec:numerics}
In this section, we conduct experiments in both combinatorial and probabilistic settings to compare the performance of DSA and the hybrid algorithm against BSA and HGBSA. Since SBSA is computationally expensive, we only conduct a small-scale experiment.

\textbf{Combinatorial setting:}
The number of infections $\nDef$ is fixed for a certain population. We conduct experiments in a population of $\nItems$. We generate $1000$ instances for each fixed number of infections $\nDef$. In each instance, we randomly select $\nDef$ people to set their status as infected. Note that HGBSA requires an estimate of upper bound of $k$. In our simulation, we use the exact $k$ as its upper bound.
\fi 

\textit{Results:}
Figures~\ref{fig:comb16:tests} and~\ref{fig:comb1024:tests} depict the performance w.r.t.\ the average number of \textit{tests}: 
We observe that SBSA and HGBSA outperform all other algorithms, while BSA has mixed performance characteristics---it offers marginal benefits over SBSA when $\nDef$ is low, but the number of tests quickly increases as $\nDef$ becomes larger. The latter is expected, as the performance of BSA is linear to the number of infections (see~\Cref{sec:prelim}).
DSA  largely outperforms BSA and performs order-optimally w.r.t.\ HGBSA (having a degradation of about $20\%$), while our hybrid variant offers further benefits compared with DSA  as $\nDef$ grows large. 

Things seem to get reversed in Figures~\ref{fig:comb16:stages} and \ref{fig:comb1024:stages} that show our results w.r.t.\ the average number of \textit{stages}:
DSA  significantly outperforms all algorithms, except for HGBSA whenever $\nDef \geq (\nItems + 1)/2$. This is because, in that regime, HGBSA applies individual testing that can occur in one stage; however, to accomplish this, HGBSA requires the exact knowledge of $\nDef$. 
Similarly, our hybrid variant performs better than prior approaches for $\nDef < \nItems/2$, while it matches the performance of DSA  as $\nDef$ goes close to $\nItems$.
For all algorithms, there seems to be an interesting trade-off between the tests and the stages (that is perhaps worth analyzing) and DSA  does keep a good balance in this regard.

% In conclusion, diagonal algorithm achieves a trade-off between timeliness and economy for a population with no prior information about infections. The hybrid algorithm saves tests compared with the diagonal algorithm. It is more timely than HGBSA. Hence, we empirically justify the efficiency of hybrid algorithm.

\begin{figure}[t]
\label{fig:comb1024}
\begin{subfigure}{\linewidth}
\centering
\includegraphics[height = 5cm, width = 0.9\linewidth]{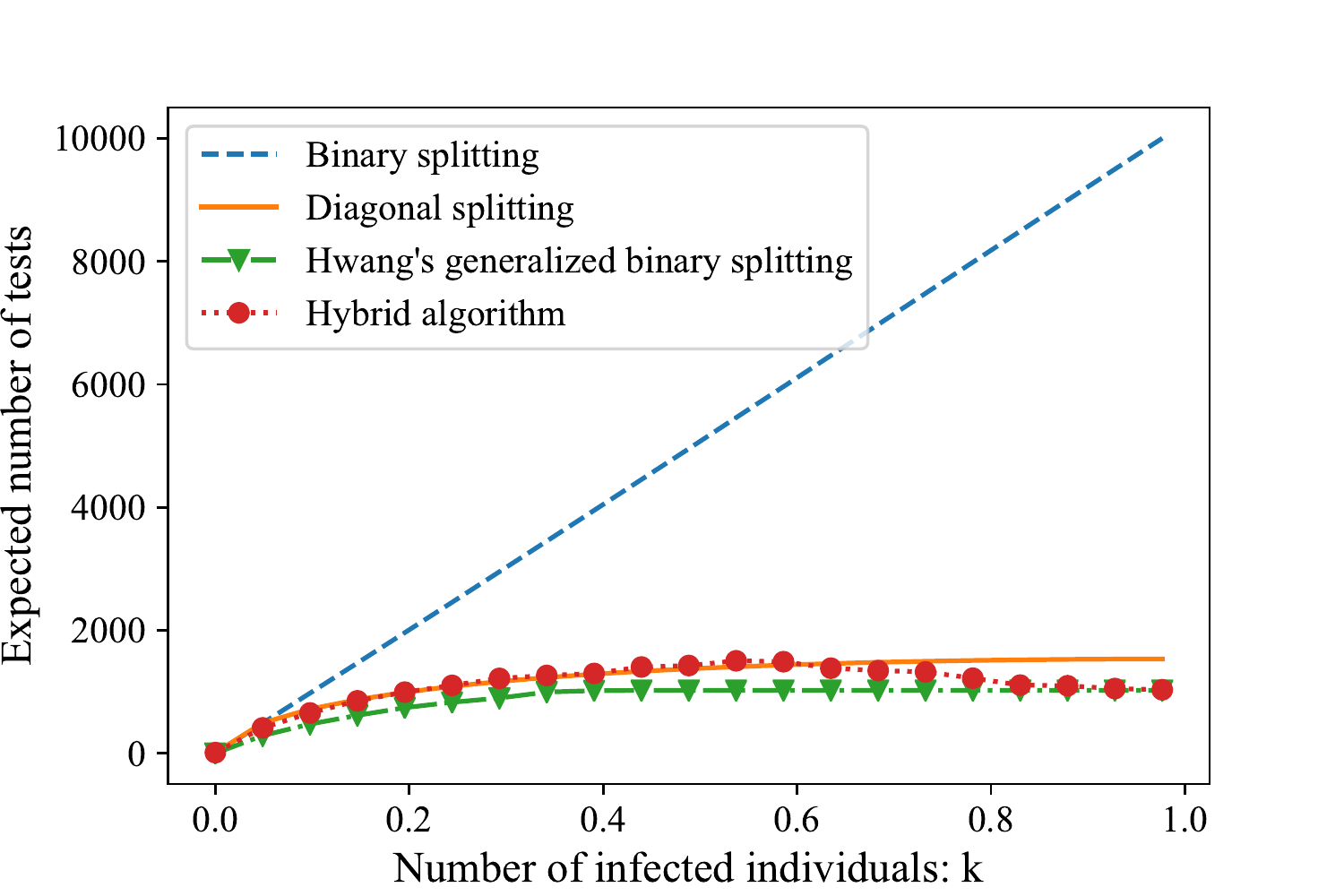}
\caption{Number of tests}
\label{fig:comb1024:tests}
\end{subfigure}

\begin{subfigure}{\linewidth}
\centering
\includegraphics[height = 5cm, width = 0.9\linewidth]{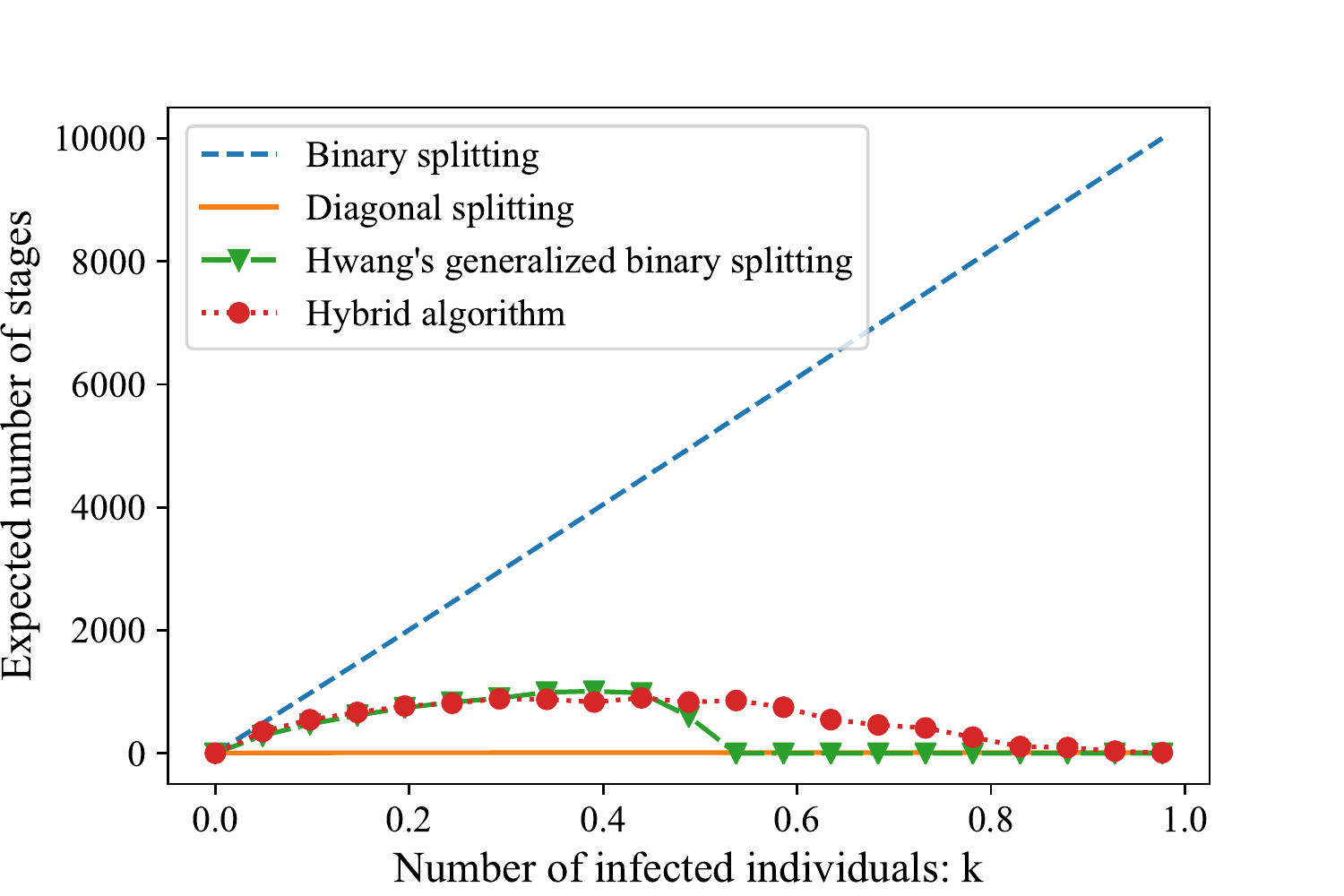}
\caption{Number of stages}
\label{fig:comb1024:stages}
\end{subfigure}%
\caption{Case $\nItems=1024$, results averaged over $1000$ instances.}
\vspace{-0.5cm}
\end{figure}

\section{Conclusion}
In this paper, we present a diagonal splitting algorithm (DSA) for adaptive group testing with no prior information about the number or the probability of infection(s). We theoretically analyze and empirically evaluate its performance, and we show that it operates well in terms of both the number of tests and the number of stages, thereby keeping a good balance between testing economy and timeliness. In our opinion, it would be interesting to further explore the trade-off between the number of stages and the number of tests, potentially by adaptively updating the estimate of $\nDef$ along the course of the stages and utilizing it efficiently.
% \textcolor{cyan}{Not clear if there is an optimal number of stages - we may want to say something like - it would be interesting to further explore  possible trade-offs between number of stages and number of tests, potentially by ***}
%%%%%%
%% To balance the columns at the last page of the paper use this
%% command somewhere at the top of the first column of the last page:
%%
% \enlargethispage{-5cm} 
%%
%% where the exact amount of page reduction has to be adapted to the
%% actual situation.
%%
%% If the balancing should occur in the middle of the references, use
%% the following trigger:
%%
% \IEEEtriggeratref{3}
%%
%% which triggers a \newpage (i.e., new column) just before the given
%% reference number. Note that you need to adapt this if you modify
%% the paper. The "triggered" command can be changed if desired:
%%
% \IEEEtriggercmd{\enlargethispage{-20cm}}
%%
%%%%%%

%%%%%%
%% References:
%% We recommend the usage of BibTeX:
%%
%\bibliographystyle{IEEEtran}
%\bibliography{definitions,bibliofile}
%%
%% where we here have assume the existence of the files
%% definitions.bib and bibliofile.bib.
%% BibTeX documentation can be obtained at:
%% http://www.ctan.org/tex-archive/biblio/bibtex/contrib/doc/
%%%%%%
%% Or you use manual references (pay attention to consistency and the
%% formatting style!):

\clearpage
\bibliographystyle{IEEEtran}
\bibliography{bibliography}

%%%%%% 
%% Appendix:
%% If needed a single appendix is created by
%%
\newpage
\appendix
%%
%% If several appendices are needed, then the command
%%
% \appendices
%%
%% in combination with further \section-commands can be used.
%%%%%%

%The appendix (or appendices) are optional. For reviewing purposes
%additional 5~pages (double-column) are allowed (resulting in a maximum
%grand total of 10~pages plus one page containing only
%references). These additional 5~pages must be removed in the final
%version of an accepted paper.

\subsection{A larger testing tree structure}
\label{app:tree16}
\begin{figure*}[t]
\begin{center}
\includegraphics[width=\textwidth]{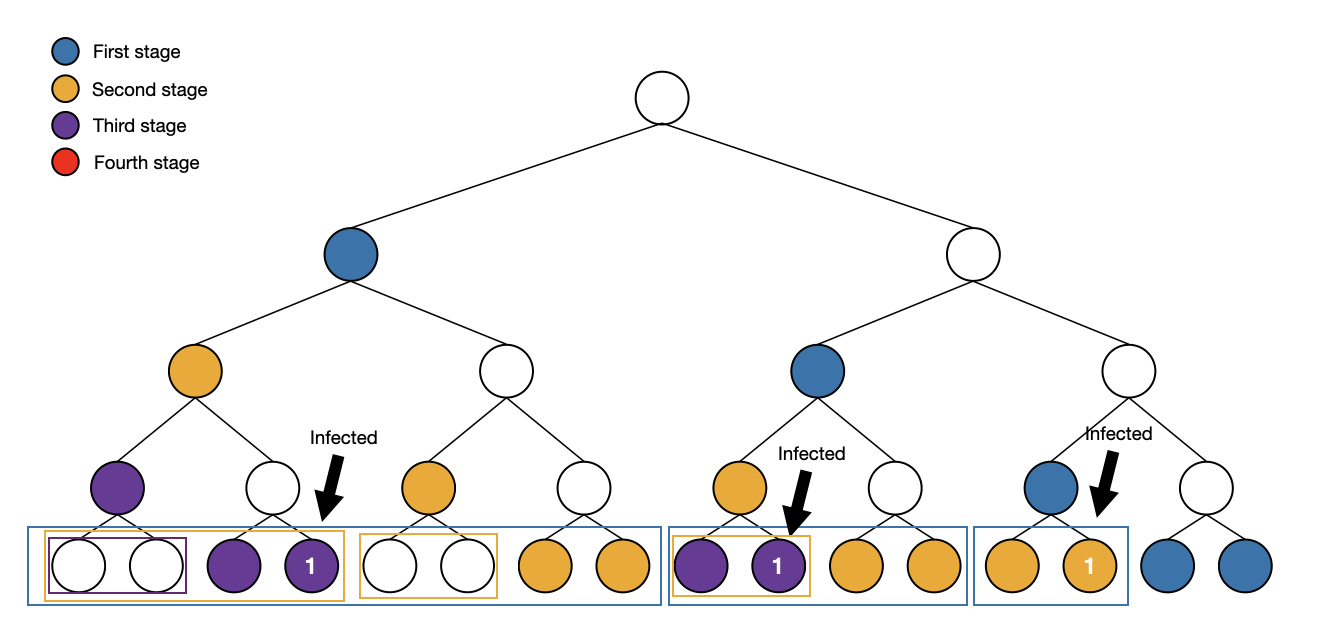}
\end{center}
\caption{An illustration of DSA  for $\nItems = 16$. The colored nodes correspond to the pooled tests of each stage.}
\label{fig:diag-demo16}
\end{figure*} 
Figure \ref{fig:diag-demo16} is an illustration of diagonal splitting for a population of 16 with 3 infections. The colored nodes correspond to the pooled tests of each stage.

\subsection{Order-optimality of the number of tests}
\label{app:optimality}
\begin{claim}
The expected number of tests $\mathbb{E}[T]$ for probabilistic model is upper bounded as follows:
$\expect[\nTests] \leq \frac{1}{C} \log_{2}{\binom{n}{k}}+(\log_{2}C+2)k$, for some constant $C$.
\end{claim}
\begin{proof}
Choose a depth $\beta$ s.t. $\epsilon = 1-e^{-\frac{k}{2^\beta}} > 0$ is a small positive constant. Then, 
\begin{equation}\beta = \log_2\frac{k}{-\ln{(1-\epsilon)}} \label{eq:beta}.\end{equation}
    From \Cref{lem:tests}, we have: 
    \begin{align}
        &\expect[\nTests] = 1+d + \sum_{i=1}^{d-1} 2^{i-1} \cdot  \prob_{i,Prob}^{+} \cdot \left(d-i+1\right) \nonumber\\         
        &= 1+d + \sum_{i=1}^{d-1} 2^{i-1} \cdot  (1- (1-\frac{k}{n})^{2^{d-i}}) \cdot \left(d-i+1\right) \nonumber \\
        &\overset{(a)}{\sim} 1+d + \sum_{i=1}^{d-1} 2^{i-1} \cdot  (1- e^{-\frac{k}{2^{i}}}) \cdot \left(d-i+1\right) \nonumber\\
        &\overset{(b)}{\leq} 1+d + 
\underbrace{\sum_{i=1}^{\beta} 2^{i-1} \cdot \left(d-i+1\right)}_{\text{(c)}} +\underbrace{\epsilon \sum_{i=\beta+1}^{d-1} 2^{i-1} \cdot \left(d-i+1\right)}_{\text{(d)}} \label{eq:upper_bound},
\end{align}
where:\\
(a) holds because $\lim_{\nItems\rightarrow \infty} (1-\frac{\nDef}{\nItems})^{r} = e^{-\frac{\nDef r}{\nItems}}$ and $\nItems = 2^{d}$; \\
(b) holds because  $1-e^{-\frac{\nDef}{2^i}} \leq 1$, and $1-e^{-\frac{\nDef}{2^i}} \leq \epsilon$ for $i > \beta$; the latter holds because $1-e^{-\frac{\nDef}{2^i}}$ is monotonically decreasing w.r.t.\ $i$. 

Next, we compute the two terms $(c)$ and $(d)$, as follows: 
  \begin{equation}
            (c) = (d+1)\sum_{i=1}^{\beta} 2^{i-1} - \underbrace{\sum_{i=1}^{\beta}i 2^{i-1}}_{\text{(e)}}  
      = (d+1)(2^{\beta}-1) - (e). \label{eq:term_c_intermediate}
  \end{equation}
  and
  \begin{align}
      (e)&=\sum_{i=1}^{\beta} \frac{\partial }{\partial x}x^{i}\bigg{|}_{x=2}= \frac{\partial }{\partial x} \sum_{i=1}^{\beta} x^{i} \bigg{|}_{x=2} \nonumber \\
      &=\frac{\partial}{\partial x} \left( \frac{1-x^{\beta+1}}{1-x} -1 \right) \bigg{|}_{x=2} \nonumber\\
      &= \frac{x^\beta (\beta x - \beta -1) + 1}{(x-1)^{2}} \bigg{|}_{x=2} \label{eq:term_e_intermediate} \\      
      &= (\beta-1)2^{\beta}+1 \label{eq:term_e},
  \end{align} by the linearity of operations.

  Similarly, 
  \begin{align}
  (d) &\overset{(f)}{=}  \epsilon \sum_{j=2}^{d-\beta} j2^{d-j} \nonumber \\
  & = \frac{1}{2}\epsilon \cdot 2^{d} \sum_{j=2}^{d-\beta} j \left(\frac{1}{2}\right)^{j-1} \nonumber \\
  & = \frac{1}{2}\epsilon \cdot2^{d} \left( -1 +  \sum_{j=1}^{d-\beta} j \left(\frac{1}{2}\right)^{j-1}\right) \nonumber \\
  & \overset{(g)}{=} \epsilon \cdot 2^{\beta}(\beta-d-2)+\frac{3}{2}\epsilon \cdot 2^d \label{eq:term_d}
  \end{align}
  where:\\ 
  (f) holds, if we set $j := d-i+1$; \\
  (g) is obtained from \cref{eq:term_e_intermediate} after substituting $\beta$ with $d-\beta$ and by taking $x=\sfrac{1}{2}$.

  We now combine (\ref{eq:upper_bound}), (\ref{eq:term_c_intermediate}), (\ref{eq:term_e}) and (\ref{eq:term_d}), and we get (as $\nItems$ goes to infinity): 
  \begin{align}
  \expect&[\nTests] \leq \left(1 -\epsilon \right)\cdot 2^{\beta}\left(d -\beta +2 \right) + \frac{3}{2}\epsilon \cdot 2^d -1 \\
  & < 2^{\beta}\left(d -\beta +2 \right) + \frac{3}{2}\epsilon \cdot 2^d -1 \\
  &\overset{(\ref{eq:beta})}{=} \frac{k}{-\ln{(1-\epsilon)}}\left(\log_{2} n-\log_{2}{\frac{k}{-\ln{(1-\epsilon)}}}+2\right)+ \frac{3}{2}\epsilon  n -1\nonumber\\
  &= \frac{k}{C} \log \frac{n}{k}+\frac{\log_{2}C+2}{C}k+ \frac{3}{2}\epsilon  n -1\nonumber\\
  &\leq \frac{1}{C} \log{\binom{n}{k}}+(\log_{2}C+2)\frac{k}{C}+ \frac{3}{2}\epsilon  n -1, \label{eq:upper_bound_final}
\end{align}
  where $C=-\ln{(1-\epsilon)}$, and inequality (\ref{eq:upper_bound_final}) holds because of the lower bound of the binomial coefficient: $\binom{\nItems}{\nDef} \geq \left(\frac{\nItems}{\nDef}\right)^{\nDef}$. 
  
  Since $\epsilon$ is an arbitrarily small constant used by our upper-bounding technique, we can always select $\epsilon = o(1/\nItems)$, so that  $\expect[\nTests] < \frac{1}{C} \log{\binom{n}{k}}+(\log_{2}{C}+2)\frac{k}{C}+ o(1)$, which gives an overall asymptotic behavior of $\expect[\nTests] = O(\log_2{\binom{\nItems}{\nDef}} + \nDef)$, where the hidden term $\frac{1}{C}$ may be large. 
\end{proof}

\subsection{An example of occurrence matrix}
\label{app:exampleM}
Figure \ref{fig:m8} is an example of an occurrence matrix $\mat{M}$ for a population size of 8. 
The rows correspond to the number of infections, while the columns correspond to the binary representations of the test outcomes of the first stage of \Cref{alg:diag}.
Note that a test outcome pattern $\vct{s}$ is valid for $\nDef$ infections only when it is possible for the pattern $\vct{s}$ to occur given $\nDef$ infections---for invalid patterns we use a $0$.

\begin{figure*}[t]
\begin{center}
\includegraphics[width=\textwidth]{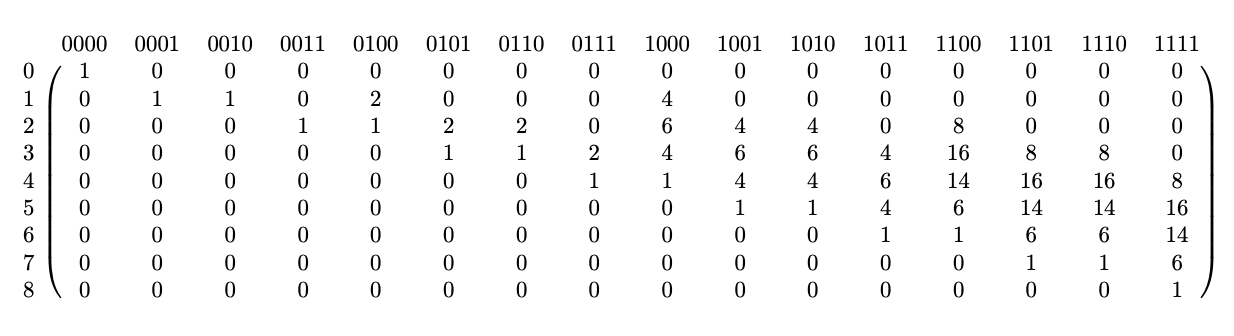}
\end{center}
\caption{An example of population size 8.}
\label{fig:m8}
\end{figure*} 

\subsection{Proof of \Cref{lem:recursiveM}}
\label{app:recursiveM}
The recursive relation for $\mat{M}$ is:
\begin{equation*}
   \mat{M}[k,s] = 
    \sum_{i=1}^{k}\binom{N_0}{i}\mat{M}[k-i,s_0],
\end{equation*} where $\mat{M}[0,0]=1$.

\begin{proof}
Note that we divide all the individuals into two categories based on a given pattern $\vct{s}$: in the first positive sub-tree or not, and then count all the cases. Note that there are three sources of possibilities: the number of infections in the first positive sub-tree, the infectious status in the first positive sub-tree, and the infectious status in the modified patterns and the corresponding number of infections. Given that there are $i$ infections in the first positive sub-tree, there are $\binom{N_0}{i}$ possible infectious status in the first positive sub-tree and $\mat{M}[k-i,s_0]$ possible infectious status for other leaves. In total, there are $\binom{N_0}{i}\mat{M}[k-i,s_0]$ possible infectious status for a given $i$. Furthermore, $i$ has $\nDef$ possibilities. The invalid pattern $\bf{s}$ for $\nDef$ will get a $0$ by the formula. Combining all yields the result.
\end{proof}

\subsection{Results for the Probabilistic case}
\label{app:probabilistic}
In the probabilistic setting, we range $\defProb$ in $[0,1]$ keeping it identical and independent for each individual,
%We set $p=\frac{k}{\nItems}$ for different $\nDef$ values. 
and we average our results over $m = \{500,1000\}$ instances for each fixed $p$.
Our results are shown in Figures~\ref{fig:prob16} and~\ref{fig:prob1024}, and as expected, they are similar to the ones of the combinatorial case.

Note that HGBSA requires the exact knowledge of the number of infected individuals $k$, which is unknown in the probabilistic setting. For that reason, the input of HGBSA is merely the average number of infections $\hat{k} = \defProb \nItems$.
\begin{figure}[t]
\begin{subfigure}[h]{\linewidth}
\includegraphics[height = 5cm, width=0.96\linewidth]{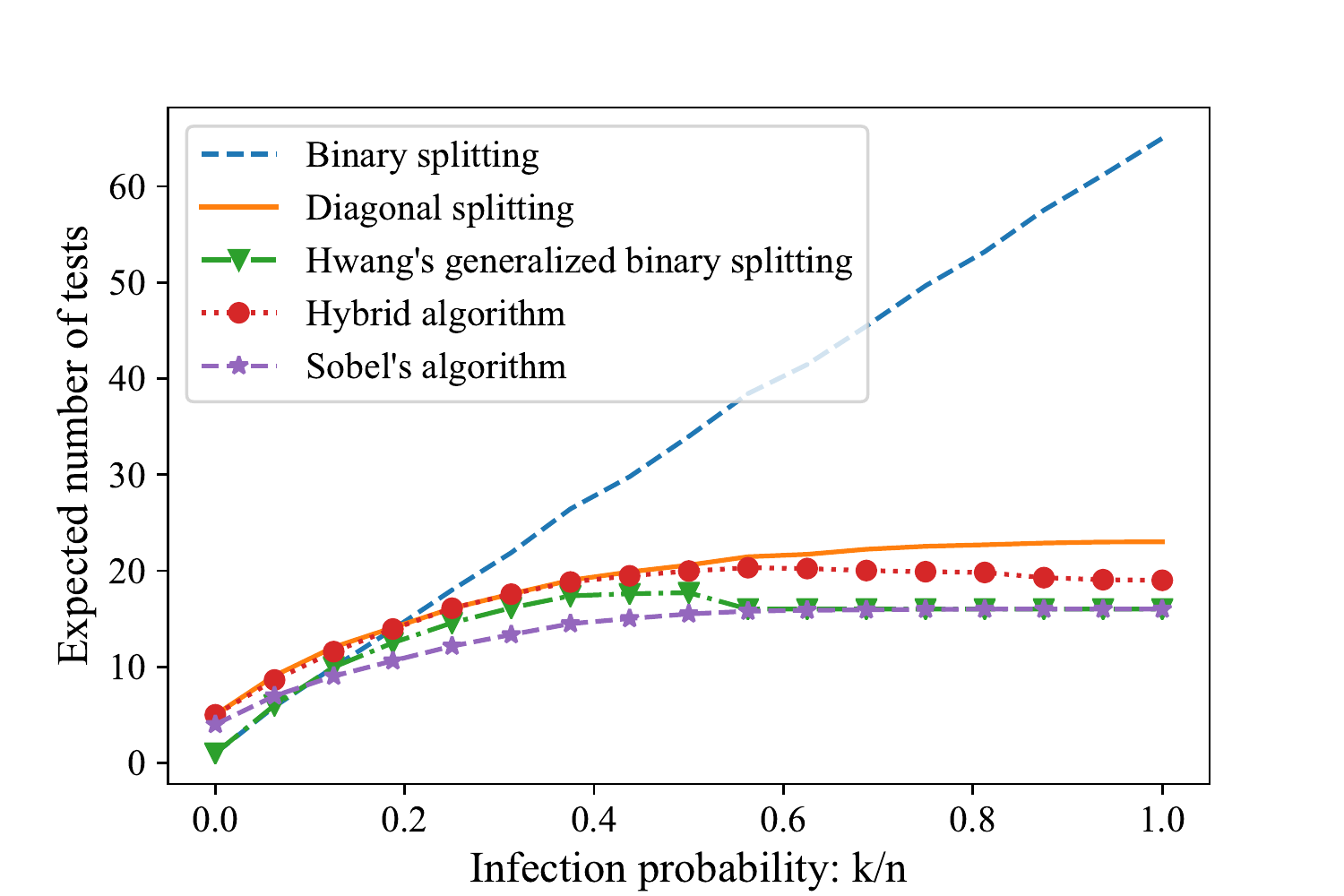}
\caption{Number of tests}
\end{subfigure}

\begin{subfigure}[h]{\linewidth}
\includegraphics[height = 5cm, width=0.96\linewidth]{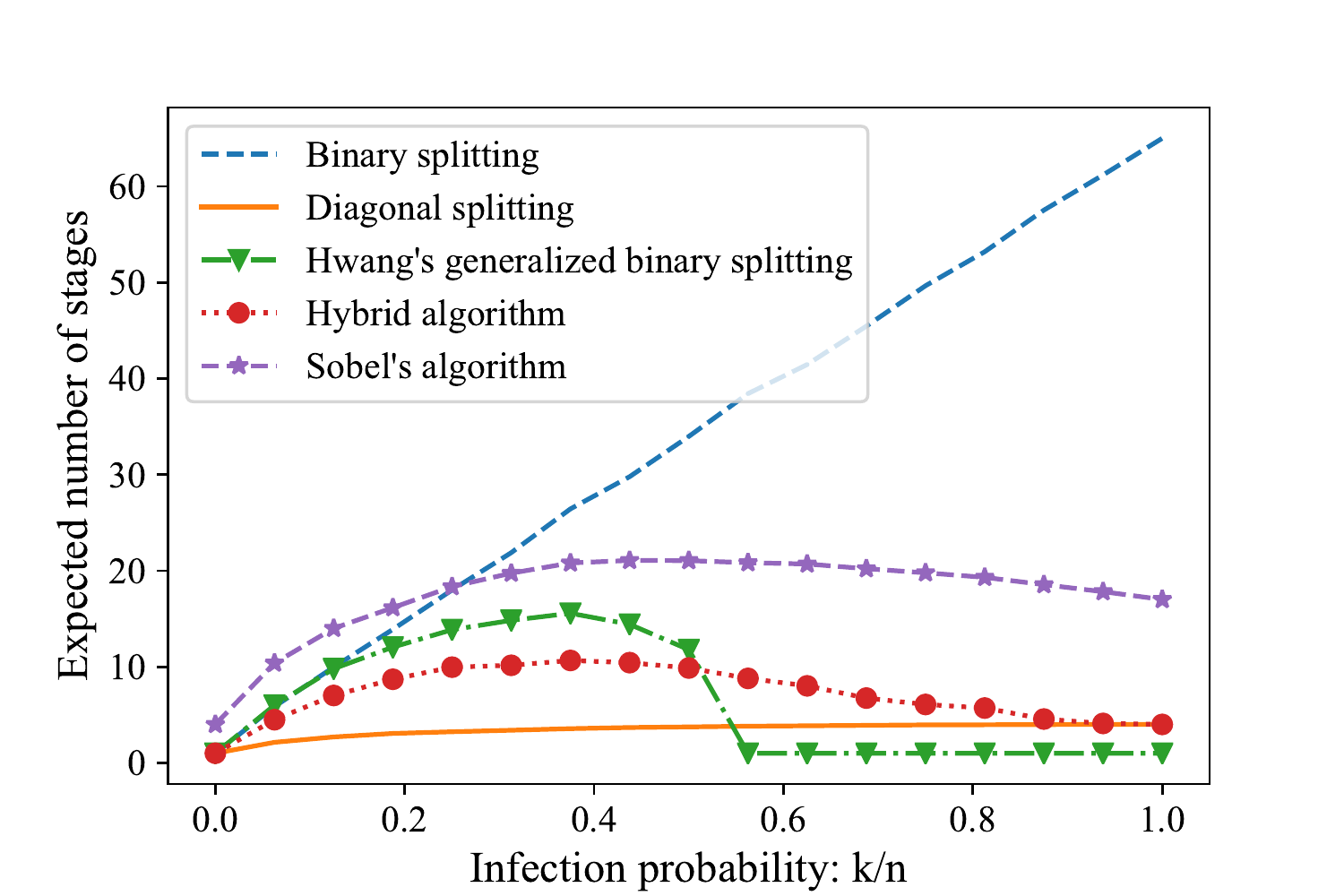}
\caption{Number of stages}
\end{subfigure}%
\caption{Probabilistic model with $\nItems=16$, results averaged over $500$ instances.}
\label{fig:prob16}
\end{figure}

\begin{figure}[t]
\begin{subfigure}[h]{\linewidth}
\includegraphics[height = 5cm, width=0.96\linewidth]{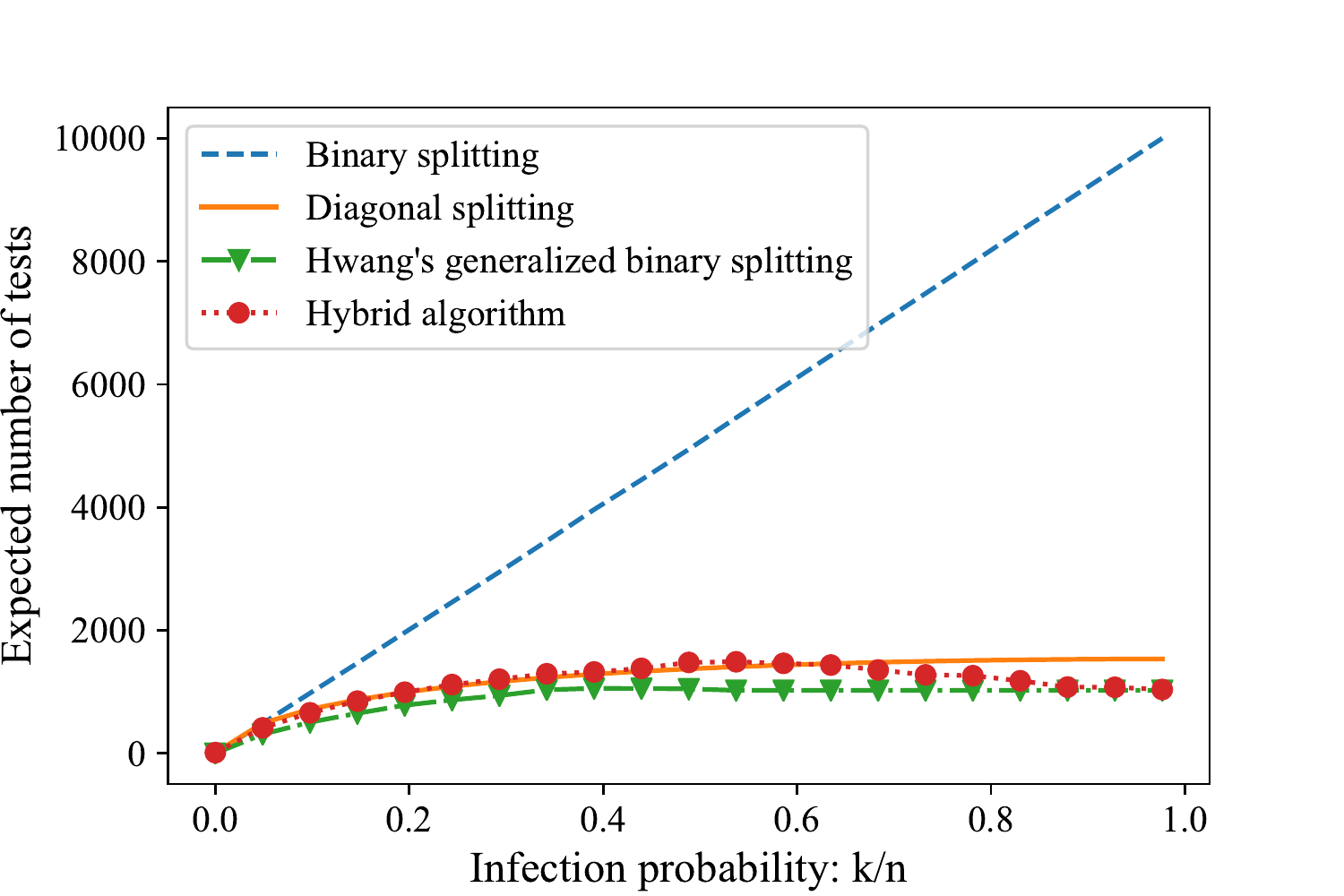}
\caption{Number of tests}
\end{subfigure}

\begin{subfigure}[h]{\linewidth}
\includegraphics[height = 5cm, width=0.96\linewidth]{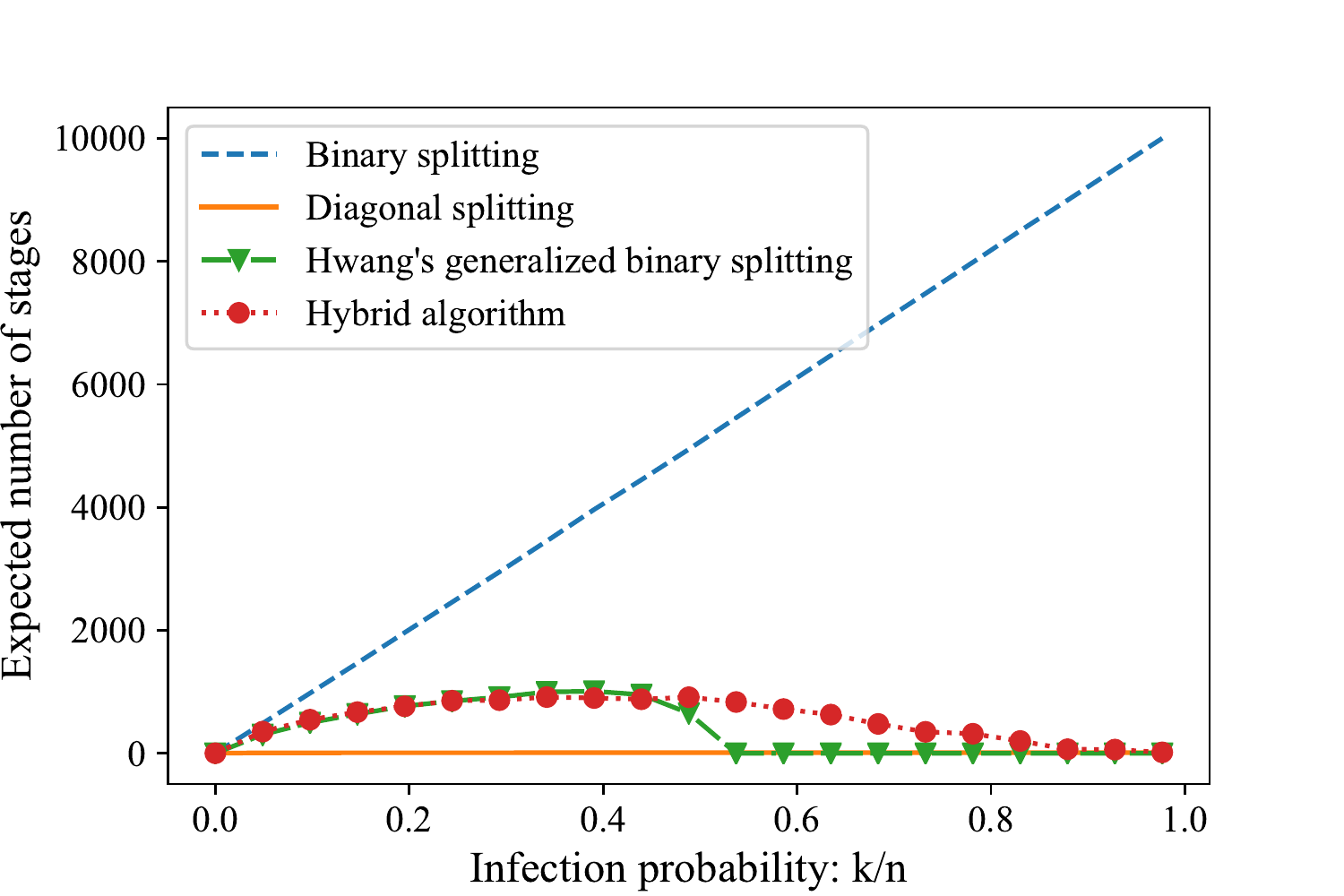}
\caption{Number of stages}
\end{subfigure}%
\caption{Probabilistic model with $\nItems=1024$, results averaged over $1000$ instances.}
\label{fig:prob1024}
\end{figure}

\end{document}